\documentclass{article}
\usepackage{ijcai16}
\usepackage{times}

\usepackage{amsthm}
\usepackage[boxed]{algorithm2e}
\usepackage{amsmath}
\usepackage{graphicx}
\usepackage{tikz}
\usetikzlibrary{shapes,snakes,positioning}
\usetikzlibrary{automata, positioning}
\usepackage{caption}

\usepackage{boxedminipage}

    \usetikzlibrary{arrows}
    \tikzstyle{vertex}=[circle,fill=black!25,minimum size=20pt,inner sep=0pt]
    \tikzstyle{edge} = [draw,thick,->]
    
	\newcommand{\citet}[1]{\citeauthor{#1}~\shortcite{#1}}
	\newcommand{\citep}{\cite}

\usepackage{booktabs}  
\usepackage{tikz}
\usepackage{verbatim}
\usepackage{boxedminipage}
\usepackage{graphicx}
\usepackage{subcaption}
\usepackage{csquotes}
\usepackage{enumitem}
\usepackage{nicefrac}
\usepackage{amsmath}

\newtheorem{theorem}{Theorem}%
\newtheorem*{theorem*}{Theorem}
\newtheorem{lemma}{Lemma}%
\newtheorem{corollary}{Corollary}%
\newtheorem{example}{Example}%

\newcommand{\ISG}{ISG\xspace}
\newcommand{\ISGs}{ISGs\xspace}
\newcommand{\task}{service\xspace}

\newcommand{\tasks}{services\xspace}
\newcommand{\Tasks}{Services\xspace}

\title{Interdependent Scheduling Games}
\author{
Andres Abeliuk \\ Data61/NICTA \\ andres.abeliuk@data61.csiro.au \And
Haris Aziz \\ Data61/NICTA and UNSW \\ haris.aziz@data61.csiro.au \And
Gerardo Berbeglia \\ University of Melbourne \\  g.berbeglia@mbs.edu
\AND
Serge Gaspers \\ UNSW and Data61/NICTA \\ sergeg@cse.unsw.edu.au \And
Petr Kalina \\ Czech Technical University \\ petr.kalina@fel.cvut.cz \And
Nicholas Mattei \\ Data61/NICTA and UNSW \\ nicholas.mattei@data61.csiro.au
\AND
Dominik Peters \\ University of Oxford \\  dominik.peters@cs.ox.ac.uk \And
Paul Stursberg \\ Technische Universit{\"a}t M{\"u}nchen \\ paul.stursberg@ma.tum.de
\AND
Pascal Van Hentenryck \\ University of Michigan \\ pvanhent@umich.edu \And
Toby Walsh \\ UNSW and Data61/NICTA \\ tw@cse.unsw.edu.au
}

\begin{document}
    
\maketitle

\begin{abstract}
We propose a model of interdependent scheduling games
in which each player controls a set of \tasks that they schedule independently.
A player is free to schedule his own \tasks at any time; however, 
each of these \tasks only begins to accrue reward for the player when all predecessor \tasks,
which may or may not be controlled by the same player, have
been activated.  This model, where players have \emph{interdependent}
\tasks, is motivated by the problems faced in planning and coordinating 
large-scale infrastructures, e.g., restoring electricity and gas to residents
after a natural disaster or providing medical care in a crisis when
different agencies are responsible for the delivery of staff, equipment, and
medicine. We undertake a game-theoretic analysis of this setting and in particular consider the issues of welfare maximization, computing best responses, Nash dynamics, and existence and computation of Nash equilibria.

\end{abstract}

\section{Introduction}

Restoring critical infrastructure in the aftermath of natural disasters or extreme weather events where water, power, and gas services may all be interrupted is one of the most important ways of limiting the impact of the disaster on society.  Our motivation for this work is drawn from situations where companies and governments need to restore \emph{interdependent infrastructure} after major disruptions due to disasters and other forces. For instance, the electric company may be able to restore power lines to individual homes, but no electricity will flow until the gas
company can supply gas to the main generator. Once the power
is flowing, the electric company receives its reward (income) from those customers
receiving power. In order to pump water, power needs to have been restored and the water lines need to be repaired. Each of these objectives are typically broken down into smaller tasks that restore availability to a subset of customers. In these settings, multiple agents (also called players) are responsible for different services and may have conflicting interests: the power company may deploy its services in an order that maximizes reach to its subscriber base first, as opposed to undertaking repairs that allow another company to restart the water pumps.
This paper formalizes a novel abstract model of this setting and studies the problem of finding a \emph{joint deployment schedule} of services through a game theoretic lens, as players in this setting are independent decision makers. We consider classic questions such as welfare maximization, best responses, and the existence and computation of Nash equilibria.

From the community's perspective, the overall goal is to reduce the size and length of the blackout.  Indeed, governments in the US plan for infrastructure restoration at a higher level than the individual company, e.g., the state government or regional emergency management planning.  However, when disasters become too large or individual companies refuse to cooperate with regional disaster management plans then
companies might be unable (or unwilling) to obey global welfare considerations in restoring their infrastructure.
\citet{cavdaroglu2013integrating} and \citet{coffrin2012last} provide models that integrate the restoration planning and scheduling decisions to show that there is significant value in this integration as opposed to tackling both problems in a decentralized manner. 
Our model of \emph{interdependent scheduling games} (\ISGs) is a step towards understanding the impact of decentralized decision making in settings with interdependencies.
Other examples of \ISGs include coordinating multiple providers for humanitarian logistics over multiple regions, where roads need to be repaired before supplies can be delivered and tents must be erected before supplies can be distributed, or the coordination of interdependent supply chains which may involve ports, terminals, railway, and truck operators \cite{van2010strategic,simon2012randomized}.

In our formalization, we consider a set of players, each of which has a set of \tasks under their control that need to be deployed.
The individual players' \tasks may have dependencies among each other and, crucially, may also be dependent on the status of other players' \tasks.
In contrast to most traditional scheduling settings, where a task cannot be scheduled unless all of its dependencies have been fulfilled, \tasks in our setting can be deployed at any time, even before its dependencies have been deployed.
However, a player only starts \emph{accruing reward} for a \task $v$ once all of its dependencies have been deployed as well.
At this point, we say that $v$ has been \emph{activated} and the player continues to gather reward for every time step in which the \task is active.
A typical reward in our setting would be collecting fees from utility subscribers who have had their service restored.  

\textbf{Contributions.}\;
We present a scheduling model with dependencies among \tasks that is suitable for scenarios in power restoration after natural disasters.
We show that when there is only a single player, a welfare-maximizing schedule can be found in polynomial time. For more players, welfare maximization becomes NP-complete even with just two \tasks per player.
Regarding game-theoretic solution concepts, we prove that in general, pure Nash equilibria are not guaranteed to exist, and that it is NP-hard to decide their existence.
On the positive side, we consider a restricted setting where all \tasks have uniform (equal) reward and prove that a pure Nash equilibrium always exists and can be computed in polynomial time. Similarly, best responses can be computed efficiently but they need not converge to a Nash equilibrium, even if rewards are uniform. For the uniform rewards case, we also give bounds for the price of anarchy and the price of stability.
Further, we provide an ILP formulation of the problem and demonstrate that, for generated data, we can find welfare maximizing schedules quickly.

\section{Related Work}
The problem of finding a schedule of tasks that maximizes the reward is an important question in scheduling, a classic area of computer science with many practical and important problems. Most classical scheduling problems focus on allocating scarce resources to multiple tasks in order to maximize an objective function or minimize total time \cite{brucker2007scheduling,lee1997current}.
In contrast to most of the scheduling literature, the dependencies (or precedence constraints) between the \tasks in our model do not prevent the player from scheduling a \task before its prerequisites are fulfilled. Instead, they keep the player from receiving \emph{reward} from the \task until the prerequisites are fulfilled.

Encouraging distributed agents, each of which may be responsible for only a small piece of a larger task, to work together to solve complex problems has a rich history in artificial intelligence and multi-agent systems research. Scheduling distributed tasks in domains where agents are imbued with their own reward functions but are ultimately cooperative as they can jointly benefit from finding coordinated schedules, has been studied in a probabilistic setting by \citet{ZhangS14a}. Additionally, \emph{task oriented domains} \cite{rosenschein1994rules}, which typically involve multiple agents working together cooperatively, are a popular framework for investigating mechanisms and properties of multi-agent domains where agents either need to work together or negotiate over work to be accomplished. \citet{ZlotkinR93} formalize the notion of strategic behavior when agents negotiate in task oriented domains. They provide a characterization of the type of lies (e.g. hiding jobs) and reward functions that admit incentive compatible mechanisms for a number of classic domains, though none of these classic domains involve scheduling with dependencies.

We focus our analysis on game-theoretic issues such as best response dynamics and Nash equilibria that are keenly applicable in settings such as ours where agents, trying to maximize independent utility, may or may not have explicit incentives to cooperate towards maximizing global welfare.
Scheduling domains in which players compete for common processing resources were introduced by \citeauthor{agnetis2000nondominated}~\shortcite{agnetis2000nondominated,agnetis2004scheduling} and \citet{baker2003multiple}. The most traditional approach in multi-agent scheduling is to consider a single centralized authority optimizing the whole domain.
There have been a number of recent works focused on decentralized scheduling mechanisms.  \citet{agnetis2007combinatorial} consider auction and bargaining models, which are useful when several players have to negotiate for processing resources on the basis of their scheduling performance. Scheduling auctions typically divide the schedule horizon into time slots, and these time slots are auctioned among the players. The bargaining approach considers two players that have to negotiate over possible schedules. \citet{abeliuk2015Bargaining} consider a two-player bargaining mechanism for any setting where the reward of one player does not depend on the actions taken by the other. Their results hence apply to special instances of \ISGs with two players.
For additional literature on mechanism design for non-cooperative scheduling games see, e.g, \citet{heydenreich2007games}, \citet{christodoulou2004coordination}, and \citet{angel2006truthful}.

Another related line of research is multi-agent project-scheduling. Here, each project is composed of a set of activities, with precedence relations between the activities, and each activity belongs to an agent.  Each activity is associated with a minimum and a maximum processing time and agents have to choose a duration for all their activities. 
Compressing the duration of an activity generates a cost to the agent, and an agents' payoff  
is a fixed proportion of the total project payment, which depends on the project completion time.
 A mechanism design approach for multi-agent project-scheduling by \citet{confessore2007market} proposes a decentralized mechanism using combinatorial auctions. Recently, \citet{briand2011cooperative} took a first step in analyzing game theoretical properties such as the existence and computation of Nash equilibria as well as studying the price of anarchy in this setting. However, their setting significantly differs from that considered in this paper in that activities of the same agent can be processed in parallel and that all agents receive some fraction of the reward of a common production process. In contrast, we focus on agents involved in independent projects with separate objective functions, only related by precedence constraints between each other.
	
\section{Our Model}
 \label{setup}

A \emph{directed graph} $G$ is a pair $(V,E)$ of a finite set of \emph{vertices} $V$ and a set of directed edges $E\subseteq V\times V$ where $(u,v) \in E$ means that there is a directed edge from $u$ to $v$ in $G$. We will always assume that $G$ is \emph{acyclic}, i.e., there is no set of edges $\{(v_1,v_2),(v_2,v_3),\dots,(v_n,v_1)\} \subseteq E$.
	We say that $G$ is \emph{transitive} if $(u,v),(v,w)\in E$ implies that $(u,w)\in E$.
The \emph{transitive closure} of a graph $G = (V,E)$ is a graph $G' = (V, E')$ such that $(u, v) \in E'$ if a directed path connects $u$ and $v$ in $G$.
The \emph{in-neighborhood} of a vertex $v$  is the set of vertices with edges to $v$ and is denoted by $N_G^-(v) = \{u\in V : (u,v)\in E\}$.

An \emph{interdependent scheduling game} (ISG) with $k$ players is given by a tuple $((T_1,\dots,T_k),G,r)$.
Each player $i$ needs to schedule a set of \emph{\tasks} $T_i$, where the $T_i$ are pairwise disjoint.
We denote the set of all \tasks by $T = \bigcup_{i=1}^k T_i$.
We assume without loss of generality that $|T_1|=\dots=|T_k| = q$. 
Within $T$ there are dependencies: a \task will not activate until it and all its prerequisites are deployed. We formalize this \emph{dependency relation} as a transitive acyclic directed graph $G=(T,E)$. If $(u,v) \in E$, then \task $v$ will generate a reward only after \task $u$ has been deployed.To be precise, at each time step $t$, each player \emph{deploys} exactly one \task. In particular, we assume that every \task takes exactly one unit of time to deploy. A \task which takes longer to deploy can be represented as a series of \tasks depending on each other where only the final \task generates a reward.
For each \task $v\in T$, there is a \emph{reward} $r(v)\geq 0$, representing payment received or subscribers served in each time period that the \task is active.
We will sometimes consider the more restrictive case of \emph{uniform rewards} where for all $v \in T, r(v) = 1$.

A solution for an \ISG is a schedule of all \tasks in $T$. As rewards are non-negative, players do not have an incentive to leave a gap between the deployment of two \tasks. We can hence represent a \emph{schedule} by a tuple $\pi = (\pi_1,\dots,\pi_k)$, where each $\pi_i:T_i \rightarrow \{1,\dots,|T_i|\}$ is a permutation of the \tasks $T_i$ of player $i$. This permutation uniquely determines the schedule for player $i$ and the position of a \task in the permutation denotes the time when it is deployed.

A \task $v$ is \emph{active} during a time step if itself and all \tasks in $N_G^-[v]$ are deployed at or before that time step. We denote by $a(v)$ the time when $v$ becomes active, i.e.
$ a(v)=\max \{\pi(w): w \in \{v\}\cup N_G^-[v]\}$.
At each time step, all active \tasks $v$ generate the reward $r(v)$.
Thus, for a schedule $\pi=(\pi_1,\dots,\pi_k)$, the \emph{utility} of player $i$ is
$ R_i(\pi) = \sum_{t=1}^{q} \;\sum_{v \in T_i, t \ge a(v)} r(v)$.
The \emph{utilitarian social welfare} (or just \emph{welfare}) of a schedule $\pi$ is $\smash{\sum_{i=1}^k} R_i(\pi)$.

We graphically represent an \ISG in Example~\ref{ex:basic}. 
Player $i$'s \tasks $T_i$ form the nodes shown in the $i$th row.  The \tasks in a row, from left to right, represent player $i$'s schedule, while the label of a \task indicates its reward. For ease of presentation, we omit arrows that are implied by transitivity of the dependency relation; the full dependency graph is the transitive closure of the depicted graph.  This representation is not completely unambiguous as a \task $v$ is identified only by $r(v)$ and the edges in $N_{G}(v)$. However, while indistinguishable (subsets of) tasks may exist, these can be interchanged within any particular outcome without effect.

\begin{example}
Consider the following example.  

\begin{minipage}{0.47\linewidth}
\centering
	\scalebox{0.9}{
	\begin{tikzpicture}[]
	\node (p1) {$\pi_1:$};
	\node[vertex] (t_1_a) [right of=p1] {$10$};
	\node[vertex] (t_2_a) [right of=t_1_a] {$1$};
	\node[vertex] (t_3_a) [right of=t_2_a] {$1$};

	\node (p2) [below of=p1] {$\pi_2:$};
	\node[vertex] (t_1_b) [right of=p2] {$1$};
	\node[vertex] (t_2_b) [right of=t_1_b] {$100$};
	\node[vertex] (t_3_b) [right of=t_2_b] {$100$};

	\path[edge] (t_2_a) to [bend left=10] (t_3_a);

	\path[edge] (t_1_a) to [bend left=10] (t_1_b);
	\path[edge] (t_2_a) to [bend left=10] (t_2_b);
	\path[edge] (t_2_a) to [bend left=10] (t_3_b);
	\end{tikzpicture}
	}
\end{minipage}
\hfill
\begin{minipage}{0.47\linewidth}
\centering
	\scalebox{0.9}{
	\begin{tikzpicture}[]
	\node (p1) {$\pi'_1:$};
	\node[vertex] (t_1_a) [right of=p1] {$1$};
	\node[vertex] (t_2_a) [right of=t_1_a] {$1$};
	\node[vertex] (t_3_a) [right of=t_2_a] {$10$};

	\node (p2) [below of=p1] {$\pi'_2:$};
	\node[vertex] (t_1_b) [right of=p2] {$100$};
	\node[vertex] (t_2_b) [right of=t_1_b] {$100$};
	\node[vertex] (t_3_b) [right of=t_2_b] {$1$};

	\path[edge] (t_1_a) to [bend left=10] (t_2_a);

	\path[edge] (t_1_a) to [bend left=10] (t_1_b);
	\path[edge] (t_1_a) to [bend left=10] (t_2_b);
	\path[edge] (t_3_a) to [bend left=10] (t_3_b);
	\end{tikzpicture}
	}
\end{minipage}

\smallskip
\noindent
Both of the \tasks with reward $100$ belong to Player 2 but depend upon a \task belonging to Player 1. For schedule $\pi$, $ R_1(\pi) = 3 \cdot 10+2 \cdot 1+1=33$ as the \task with reward 10 is active for three time steps and the other \tasks are active for two and one time step, respectively. Similarly, $R_2(\pi) = 3\cdot 1+ 2\cdot 100+100=303$. For $\pi'$, $ R_1(\pi') =3\cdot 1+2\cdot 1+10=15$ while $ R_2(\pi') = 3\cdot 100+ 2\cdot 100+1=501$.  Hence, Player 1 can sacrifice some individual reward to increase welfare.
\label{ex:basic}
\end{example}

\section{Best Responses}
If all other players' actions are fixed, the resulting problem for an individual player is that of finding a best response.  Let $R_i(\pi_{-i}, \pi_i')$ be the reward for player $i$ for the schedule $(\pi_1,\dots, \pi_{i-1},\pi'_i,\pi_{i+1},\dots,\pi_k)$.

\vspace{.1cm}
\noindent
\textbf{Problem}: \textsc{ISG Best Response}. \\
\noindent
\textbf{Instance:} An \ISG $((T_1,\dots,T_k),G,r)$, a schedule $\pi_{-i}$ for all players $\{1,\dots,k\} \setminus \{i\}$, and an integer $W$.\\
\noindent
\textbf{Question:} Is there a $\pi_i'$ such that $R_i(\pi_{-i}, \pi_i')\geq  W$?
\vspace{.1cm}

Assuming that players are individually rational they will favor schedules that maximize their own reward, i.e., their own subscriber base or service network. Hence an individual player will always favor a schedule such that every \task $v$ that he controls is deployed only after all other \tasks under the player's control that $v$ depends on have been deployed.
Formally, the following Lemma holds:

\begin{lemma}\label{lem:conflict-free}
	Let $((T_1,\dots,T_k),G,r)$ be an ISG with general rewards and $\pi_{-i}$ a schedule for all players except player $i$. Let $G_i=(T_i,E_i)$ denote the subgraph of $G$ induced by the vertices in $T_i$. Then, there exists a best response $\pi_i$ for player $i$ such that
	\begin{equation}
		\pi_i(u)<\pi_i(v) \text{ for all } (u,v) \in E_i. \label{eq:conflict-free}
	\end{equation}
\end{lemma}
\begin{proof}
	Let $\sigma(\pi_i) := |\{v \in T_i: \exists (u,v) \in E_i \text{ with } \pi_i(u)>\pi_i(v)\}|$ denote the number of \tasks in $T_i$  that depend on another \task in $T_i$ which is scheduled later. Let $\pi'_i$ denote a best response of player $i$ such that $\sigma(\pi'_i)$ is minimal among all best responses. We suppose for contradiction that the statement is false, therefore $\sigma(\pi'_i) \geq 1$. Choose $(u,v) \in E_i$ with $\pi'_i(u)>\pi'_i(v)$ in such a way that there is no $u'$ with $(u',v) \in E_i$ and $\pi'_i(u') > \pi'_i(u)$. Consider the following modified schedule $\pi^*_i$ for player $i$:
	\begin{equation*}
		\pi^*_i(w) :=\begin{cases}
			\pi'_i(w)-1	& \pi'_i(w) \in [\pi'_i(v) + 1, \pi'_i(u)]\\
			\pi'_i(u)	& w=v\\
			\pi'_i(w)	& \text{else}.
		\end{cases}
	\end{equation*}

The following two properties hold:
\textbf{(i)}:\;The schedule $\pi^*_i$ is also a best response. The only \task that is scheduled to a later time in $\pi^*_i$ (and hence could cause itself or \tasks depending on it to generate a smaller reward) is $v$. However, $v$ did not activate before time step $\pi'_i(u)$ under $\pi'$ and as $\pi^*_i(v)=\pi'_i(u)$ the reward generated by $v$ does not change. The same holds for all \tasks that depend on $v$.
\textbf{(ii)}:\; $\sigma(\pi^*_i) < \sigma(\pi'_i)$. First, note that $v$ does not contribute towards $\sigma$ anymore as $\pi^*_i(v) > \pi^*_i(u)$ (the same holds for all other \tasks that $v$ depends upon by the maximality of $u$). Now, consider any \task $w$ that did not contribute to $\sigma(\pi'_i)$. As the ordering among all \tasks except $v$ remains the same, such a $w$ can only contribute to $\sigma(\pi^*_i)$ if it depends on $v$ and $\pi'_i(v) < \pi'_i(w) < \pi^*_i(v) = \pi'_i(u)$. But then, it must also depend on $u$ by transitivity and hence it contributed to $\sigma(\pi'_i)$ already.
	From (i) and (ii), we obtain a contradiction to minimality of $\sigma(\pi'_i)$, concluding the proof.
\end{proof}

Note that performing pairwise swaps in a player's scheduled services is not sufficient in the context of the above proof as this may introduce new forward edges. The above lemma holds for general rewards. If rewards are uniform, we can use Lemma~\ref{lem:conflict-free} to derive a polynomial-time algorithm for an individual player's best response to all other players' schedules. 

\begin{theorem}\label{th:br-uniform}
For an \ISG with uniform rewards, there exists a polynomial-time algorithm to compute a best response.
\end{theorem}
\begin{proof}
	 Consider the subgraph $G_i$ of $G$ induced by the set $T_i$ of \tasks belonging to player $i$. For every \task $u$, denote by $\eta(u)$ the lower bound on its activation time imposed by $\pi_{-i}$. Formally, $\eta(u):= \max\{\pi(v): v \in T\setminus T_i, (v,u) \in E\}$. Note that $(u,w) \in E_i$ implies $\eta(w) \geq \eta(u)$ by transitivity of $E_i$.
	
	We give a greedy algorithm that solves the problem optimally. Starting from the first time step, the algorithm successively schedules a \task which minimizes $\eta$ among all \tasks with no incoming edges in $G_i$. Such a \task always exists, as $G$ (and hence all subgraphs) is acyclic. The \task with all its (outgoing) edges is then removed from $G_i$.
	
	To prove that the algorithm yields an optimal solution, let $\pi_i$ be the outcome of the algorithm. Suppose for contradiction that $\pi_i$ is not optimal. Let $\pi^*_i$ be an optimal schedule satisfying condition \eqref{eq:conflict-free} (which exists by Lemma~\ref{lem:conflict-free}) maximizing the first time slot for which any such schedule differs from $\pi_i$. Formally, there exists $k \in N$ such that $(\pi^*_i)^{-1}(i)=(\pi_i)^{-1}(i)$ for all $i < k$ and there is no optimal schedule $\pi'_i$ satisfying condition \eqref{eq:conflict-free} with $(\pi'_i)^{-1}(i)=(\pi_i)^{-1}(i)$ for all $i < k+1$.
	
	Let $a:=(\pi^*_i)^{-1}(k)$ and $b:=(\pi_i)^{-1}(k)$.
	Consider the subgraph of $G_i$ from which the first $k-1$ entries of $\pi_i$ (and hence of $\pi^*_i$) have been removed. First, observe that $a$ cannot have any incoming edges as $\pi^*_i$ satisfies condition \eqref{eq:conflict-free}. Hence, it holds that $\eta(b) \leq \eta(a)$, otherwise the algorithm would have selected $a$ rather than $b$.
	We distinguish three cases:
	\begin{enumerate}[leftmargin=*]
		\item $\eta(b) \leq \pi^*_i(a)$. In this case, we set
		\begin{equation*}
			\pi^{**}_i(w) :=\begin{cases}
				\pi^{*}_i(w)+1	& \pi^*_i(w) \in [\pi^*_i(a),\pi^*_i(b)-1]\\
				\pi^{*}_i(a)	& w=b\\
				\pi^{*}_i(w)	& \text{else}.
			\end{cases}
		\end{equation*}
		The reward generated by $b$ increases by $\pi^*_i(b)-\pi^*_i(a)$, at the same time the reward of at most $\pi^*_i(b)-\pi^*_i(a)$ \tasks decreases by 1. Hence, $\pi^{**}$ is still optimal and satisfies condition \eqref{eq:conflict-free}. Furthermore, $(\pi_i)^{-1}(k)=(\pi^{**}_i)^{-1}(k)$, contradicting $\pi^*$'s maximality.
		\item $\pi^*_i(a) < \eta(b) < \pi^*_i(b)$. We construct a new schedule $\tilde\pi^*_i$ with $\tilde\pi^*_i(b)=\eta(b)$ as above. Then, proceed as in 3.
		\item $\pi^*_i(b) \leq \eta(b)$. Construct a schedule $\pi^{**}_i$ as follows: Set $\pi^{**}_i(b):=\pi^{**}_i(a)$. Let $a'$ be the earliest successor of $a$. If $\pi^{*}_i(a')>\pi^{*}_i(b)$, set $\pi^{**}_i(a):=\pi^{*}_i(b)$ and $\pi^{**}_i(w):=\pi^{*}_i(w)$ for all other \tasks. Otherwise, set $\pi^{**}_i(a):=\pi^{*}_i(a')$ and let $a''$ be the earliest successor of $a'$. Proceed with $a''$ (and possibly its earliest successor) as above until an earliest successor $a^*$ satisfies $\pi^{*}_i(a^*)>\pi^{*}_i(b)$. The resulting schedule $\pi^{**}$ is still optimal and satisfies condition \eqref{eq:conflict-free}. Furthermore, $(\pi_i)^{-1}(k)=(\pi^{**}_i)^{-1}(k)$, contradicting $\pi^*$'s maximality.
	\end{enumerate}

	In all three cases, we reach a contradiction which proves that our assumption was wrong and $\pi_i$ is indeed optimal.
\end{proof}

In contrast, we can obtain the following statement about general rewards by reduction from single-player welfare maximization using Theorem~\ref{th:one-player-np}.

\begin{corollary}
For an \ISG with general rewards, computing a best response is NP-complete.
\end{corollary}

\section{Welfare Maximization}
\label{sec:welfare}

A central planner would want to find a schedule that maximizes the \emph{welfare}, i.e., 
the most profitable \tasks in $T$ activated for the longest amount of time.

\vspace{.1cm}
\noindent
\textbf{Problem}: \textsc{ISG Welfare}. \\
\noindent
\textbf{Instance:} An \ISG $((T_1,\dots,T_k),G,r)$ and an integer $w$. \\
\noindent
\textbf{Question:} Is there a $\pi$ such that $\sum_{i=1}^k R_i(\pi)\ge w$?
\vspace{.1cm}

Intuitively, it might seem desirable to design schedules where no \task has to wait for its activation after it has been deployed. We call such schedule \emph{conflict-free}. For uniform rewards, if a conflict-free schedule exists then every welfare-maximizing schedule obviously has to be conflict-free. A similar statement holds for single-player games by the construction of $\pi^*_i$ in Lemma~\ref{lem:conflict-free} (proof omitted).

\begin{theorem}\label{thm:topological}
	For one player and general rewards, every welfare-maximizing schedule is a conflict-free schedule.
\end{theorem}

\noindent
However, this property does not hold in the case of more than one player and general rewards.  This can be seen by considering Example~\ref{ex:basic} and making all other \tasks dependent on $\pi_1$'s \task with reward 10. Then, any conflict-free schedule will yield welfare 319 while the welfare-maximizing schedule is 417, yielding the following theorem:
\begin{theorem}\label{thm:conflictisgood}
	For multiple players and general rewards, even if a conflict-free schedule exists, the welfare-maximizing schedule(s) might not be conflict-free.
\end{theorem}
Turning to computational complexity, we observe that for one player welfare maximization is equivalent to finding a best response, hence with Theorem~\ref{th:br-uniform} we get the following.

\begin{corollary}\label{th:welfare-is-easy-1player}
	For uniform rewards, \textsc{ISG Welfare} can be solved in polynomial time for a single player.
	\end{corollary}

However, when we either increase the number of players (Thm.~\ref{th:welfare-hard-2tasks-each}) or relax the restriction of uniform rewards (Thm.~\ref{th:one-player-np}), the problem is NP-hard for surprisingly restricted cases.

			\begin{theorem}\label{th:welfare-hard-2tasks-each}
			  \textsc{ISG Welfare} is NP-complete, even when the rewards are uniform and each player has two \tasks.
			\end{theorem}
			\begin{proof}
			 The problem is in NP since we can efficiently compute the welfare of a given schedule. For NP-hardness, we reduce from \textsc{Min 2SAT}~\citep{KKM94a} which asks: Given a 2CNF formula $F$ where each clause contains exactly two literals, and an integer $k$, is there an assignment to the variables of $F$ such that at most $k$ clauses are satisfied?

			 For each variable $x$ in $F$, create a player $P_x$ with \tasks $T_x = \{x, \neg x\}$.
			 For each clause $c$ in $F$, create a player $P_c$ with \tasks $T_c = \{c_1, c_2\}$.
			 For each clause $c = (\ell_1 \vee \ell_2)$, the precedence graph contains $(c_1,c_2)$,
			 $(\ell_1,c_1)$, and $(\ell_2,c_1)$.
			 Rewards are uniform, and we set $w = 3n + 3m - k$, where $n$ and $m$ are the number of variables and clauses of $F$. 

			 It remains to prove that $F$ has an assignment satisfying at most $k$ clauses if and only if the ISG has a schedule generating a reward of at least $w$.
			 For the forward direction, suppose $F$ has an assignment $\alpha: \mbox{var}(F)\rightarrow \{0,1\}$ satisfying at most $k$ clauses.
			 Consider the schedule where, for each variable $x$, the player $P_x$ schedules first the literal of $x$ that is set to false by $\alpha$, i.e., $x$ is scheduled before $\neg x$ iff $\alpha(x)=0$.
			 Additionally, for each clause $c$, the \task $c_1$ is scheduled before $c_2$.
			 This schedule generates a reward of $3$ for each variable: a reward of $1$ at the first time step and a reward of $2$ at the second time step.
			 For a satisfied clause $c$, the schedule generates a reward of $2$: at the first time step no reward is generated since the literal satisfying the clause is scheduled at the second time step and there is an arc from that literal to $c_1$, and a reward of $2$ is generated at the second time step.
			 For an unsatisfied clause $c$, the schedule generates a reward of $3$: since neither literal satisfies the clause, both literals are scheduled at the first time step.
			 Thus, the utility generated for this schedule is at least $3n+3m-k$.

			 For the reverse direction, let $\pi$ be a schedule generating a reward of at least $w$.
			 Consider the assignment $\alpha: \mbox{var}(F) \rightarrow \{0,1\}$ with $\alpha(x)=0$ iff player $P_x$ schedules $x$ at the first time step.
			 Note that at the second time step, each player generates a reward of $2$.
			 Also, each player corresponding to a variable generates an additional reward of $1$ at the first time step since his \tasks have in-degree $0$.
			 So, at least $3n+3m-k-(3n+2m)=m-k$ additional clause players generate a reward of $1$ at the first time step.
			 But, for each such clause $c$, $c_1$ is scheduled before $c_2$ and both literals occurring in $c$ are scheduled at the first time step, which means that the assignment $\alpha$ sets these literals to false. Therefore, $\alpha$ does not satisfy $c$. We conclude that $\alpha$ satisfies at most $k$ clauses.
\end{proof}

\begin{theorem}\label{th:one-player-np}
	For general rewards, \textsc{ISG Welfare} is NP-complete even for a single player.
\end{theorem}
\noindent
The proof, omitted for space, is a reduction from the NP-hard problem \textsc{Single machine weighted completion time} \citep{lenstra1978complexity}.
It relies on Theorem~\ref{thm:topological} and an adjustment of rewards.

\subsection{Integer Programming Formulation}

 \begin{figure}[t]
 \centering
 \includegraphics[width=.99\columnwidth]{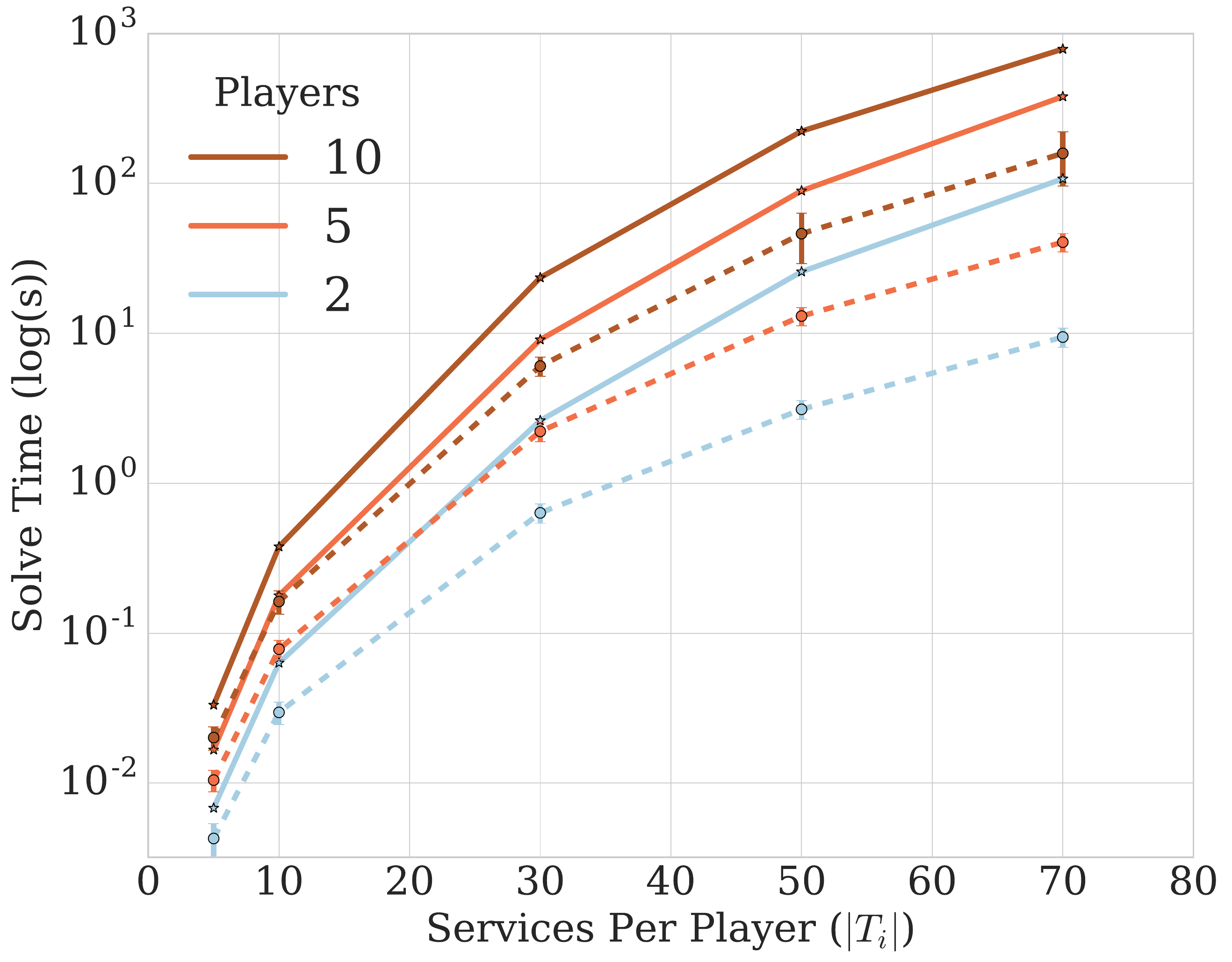}
 \caption{Mean runtime of the ILP over 1,000 random \ISG instances varying the number of players, \tasks, and reward type; error bars represent one standard deviation ($\sigma$). The solid lines are instances with general rewards, the dashed lines are instances with uniform rewards. The plot is semi-logarithmic, so a straight line represents an exponential increase in time.  For the general rewards case, error bars are not included for clarity; for $|T_i| \in \{10,30,50\}$ the numbers are small, $\sigma=30$ seconds in the worst case, however, for 70 \tasks and 10 players this balloons to 200 seconds.}
 \label{fig:results}
 \end{figure}

While the general problem of finding a welfare maximizing schedule for an \ISG instance is computationally hard, it may still be solvable for instances of moderate size.  The \textsc{ISG Welfare} problem admits a natural integer linear programming (ILP) formulation.
For each \task $v\in T$ and time step $t \in [q]$, we introduce two binary decision variables $a_{v,t}$ and $s_{v,t}$. Let $s_{v,t} = 1$ if and only if \task $v$ is scheduled at time $t$, and $a_{v,t} = 1$ if and only if \task $v$ is active at time $t$.
\vspace{-1.5mm}
\[
\begin{array}{rrr}%
\text{max} & \sum_{v\in T} \sum_{t=1}^q a_{v,t} \cdot r(v) 		& 		\\
\text{s.t.} 	& \sum_{t=1}^q s_{v,t} = 1 					& \forall v\in T \\
		& \sum_{v\in T_i} s_{v,t} = 1 					& \forall i \in [k], \forall t \in [q] \\
  		& a_{v,t} \le \sum_{t' = 1}^t s_{v,t'} 				& \forall v\in T, \forall t \in [q]\\
  		& \vphantom{\sum_t^t} a_{v,t} \le a_{w,t} 			& \forall (w,v)\in E, \forall t \in [q] \\
\end{array}
\]

We implemented the ILP and solved 1000 randomly generated instances where (a) general rewards are drawn from [50,100] and (b) rewards are uniform.  
The dependency graphs are generated by first randomly permuting the list of all \tasks; 
then for each \task $i$, drawing a random number of child \tasks $c \in \{0,1,2\}$ and adding edge
$(i, i+c)$ with probability $0.5$.  Increasing the number/likelihood of dependencies by 
increasing the potential number of children or increasing the connection probability significantly
increases runtime.
Figure~\ref{fig:results} shows the results for different parameters 
using Gurobi 6.5 on a computer equipped with an 2.0 GHz Intel Xeon E5405 CPU with 4 GB of RAM. 
The results suggest that, despite worst case hardness, the running times remain feasible, at worst $\approx$ 600s, for practically relevant problem sizes: up to 10 players with 70 \tasks each.

\section{Nash Dynamics and Equilibria} 
\label{nash}

We now turn to the situation where players may respond to each other's schedule changes. This is an important question for game theoretic analyses as it allows us to see which states leave no incentives for self-interested players to deviate; and what can happen when players are continually responding to the moves of one another. An important first question is whether a sequence of best responses terminates.

\begin{theorem}
For \ISGs with uniform rewards, best responses can cycle.
\end{theorem}
\begin{proof}
Consider the following example depicting a sequence of best responses.
Starting with the lower right schedule $\pi_D$ we move to the upper left schedule $\pi_A$
where Player 2 has changed his schedule in a best response to $\pi_D$.  We then read left
to right, top to bottom, to end up back at $\pi_D$.

\smallskip
\noindent
\begin{tabular}{l | r}
	\parbox{0.45\linewidth}{
        \centering
        \scalebox{0.75}{
        \begin{tikzpicture}[]
        \node (p1) {$\pi^A_1$:};
        \node[vertex] (t_0_0) [right of=p1]{$c$};
        \node[vertex] (t_0_1) [right of=t_0_0] {$a$};
        \node[vertex] (t_0_2) [right of=t_0_1] {$d$};
        \node[vertex] (t_0_3) [right of=t_0_2] {$b$};
        \node (p2)  [below of=p1] {$\pi^A_2$:};
        \node[vertex] (t_1_0) [below of=t_0_0] {$d$};
        \node[vertex] (t_1_1) [right of=t_1_0] {$a$};
        \node[vertex] (t_1_2) [right of=t_1_1] {$c$};
        \node[vertex] (t_1_3) [right of=t_1_2] {$b$};
        \path[edge] (t_1_2) to [bend left=15] (t_0_0);
        \path[edge] (t_0_1) to [bend left=10] (t_1_1);
        \path[edge] (t_1_0) to [bend left=15] (t_0_2);
        \path[edge] (t_0_3) to [bend left=10] (t_1_3);
        \end{tikzpicture}
        }
        $\pi^A_2$ response to $\pi^D_1$
        $R(\pi^A_1) = 8$, $R(\pi^A_2) = 10$
        \vspace{1mm}
	}
&
	\parbox{0.45\linewidth}{
        \centering
        \scalebox{0.75}{
        \begin{tikzpicture}[]
        \node (p1) {$\pi^B_1$:};
        \node[vertex] (t_0_0) [right of=p1]{$d$};
        \node[vertex] (t_0_1) [right of=t_0_0] {$b$};
        \node[vertex] (t_0_2) [right of=t_0_1] {$c$};
        \node[vertex] (t_0_3) [right of=t_0_2] {$a$};
        \node (p2)  [below of=p1] {$\pi^B_2$:};
        \node[vertex] (t_1_0) [below of=t_0_0] {$d$};
        \node[vertex] (t_1_1) [right of=t_1_0] {$a$};
        \node[vertex] (t_1_2) [right of=t_1_1] {$c$};
        \node[vertex] (t_1_3) [right of=t_1_2] {$b$};
        \path[edge] (t_1_0) to [bend left=10] (t_0_0);
        \path[edge] (t_0_1) to [bend left=15] (t_1_3);
        \path[edge] (t_1_2) to [bend left=10] (t_0_2);
        \path[edge] (t_0_3) to [bend left=15] (t_1_1);
        \end{tikzpicture}
        }
        $\pi^B_1$ response to $\pi^A_2$ 
        $R(\pi^B_1) = 10$, $R(\pi^B_2) = 8$
        \vspace{1mm}
        }
\\ \hline
	\parbox{0.45\linewidth}{
	 \vspace{0.5mm}
        \centering
         \scalebox{0.75}{
        \begin{tikzpicture}[]
        \node (p1) {$\pi^C_1$:};
        \node[vertex] (t_0_0) [right of=p1]{$d$};
        \node[vertex] (t_0_1) [right of=t_0_0] {$b$};
        \node[vertex] (t_0_2) [right of=t_0_1] {$c$};
        \node[vertex] (t_0_3) [right of=t_0_2] {$a$};
        \node (p2)  [below of=p1] {$\pi^C_2$:};
        \node[vertex] (t_1_0) [below of=t_0_0] {$c$};
        \node[vertex] (t_1_1) [right of=t_1_0] {$d$};
        \node[vertex] (t_1_2) [right of=t_1_1] {$b$};
        \node[vertex] (t_1_3) [right of=t_1_2] {$a$};
        \path[edge] (t_1_0) to [bend left=5] (t_0_2);
        \path[edge] (t_1_1) to [bend left=10] (t_0_0);
        \path[edge] (t_0_1) to [bend left=10] (t_1_2);
        \path[edge] (t_0_3) to [bend left=10] (t_1_3);
        \end{tikzpicture}
        }
        $\pi^C_2$ response to $\pi^B_1$
        $R(\pi^C_1) = 9$, $R(\pi^C_2) = 10$
       	}
&
	\parbox{0.45\linewidth}{
	 \vspace{0.5mm}
        \centering
         \scalebox{0.75}{
        \begin{tikzpicture}[]
        \node (p1) {$\pi^D_1$:};
        \node[vertex] (t_0_0) [right of=p1] {$c$};
        \node[vertex] (t_0_1) [right of=t_0_0] {$a$};
        \node[vertex] (t_0_2) [right of=t_0_1] {$d$};
        \node[vertex] (t_0_3) [right of=t_0_2] {$b$};
        \node (p2)  [below of=p1] {$\pi^D_2$:};
        \node[vertex] (t_1_0) [below of=t_0_0] {$c$};
        \node[vertex] (t_1_1) [right of=t_1_0] {$d$};
        \node[vertex] (t_1_2) [right of=t_1_1] {$b$};
        \node[vertex] (t_1_3) [right of=t_1_2] {$a$};
        \path[edge] (t_1_0) to [bend left=10] (t_0_0);
        \path[edge] (t_1_1) to [bend left=10] (t_0_2);
        \path[edge] (t_0_1) to [] (t_1_3);
        \path[edge] (t_0_3) to [bend left=10] (t_1_2);
        \end{tikzpicture}
        }
	$\pi^D_1$ response to $\pi^C_2$
	$R(\pi^D_1) = 10$, $R(\pi^D_2) = 9$
	}
\end{tabular}
\end{proof}

\subsection{\ISGs with Uniform Rewards}

A schedule $\pi$ is in \emph{pure Nash equilibrium} (PNE) if no player can obtain strictly more utility by unilaterally changing his own schedule; formally, $R_i(\pi_{-i}, \pi) \ge R_i(\pi_{-i}, \pi_i')$ for all players $i$ and all schedules $\pi_i'$ of player $i$.
For instance, note that the above example, despite having a sequence of best responses that cycle, does admit the PNE depicted below:
\begin{center}
  \scalebox{0.8}{
        \begin{tikzpicture}[]
        \node (p1) {$\pi^A_1$:};
        \node[vertex] (t_0_0) [right of=p1]{$a$};
        \node[vertex] (t_0_1) [right of=t_0_0] {$b$};
        \node[vertex] (t_0_2) [right of=t_0_1] {$c$};
        \node[vertex] (t_0_3) [right of=t_0_2] {$d$};
        \node (p2)  [below of=p1] {$\pi^A_2$:};
        \node[vertex] (t_1_0) [below of=t_0_0] {$a$};
        \node[vertex] (t_1_1) [right of=t_1_0] {$b$};
        \node[vertex] (t_1_2) [right of=t_1_1] {$c$};
        \node[vertex] (t_1_3) [right of=t_1_2] {$d$};
        
        \path[edge] (t_0_0) to [bend left=10] (t_1_0);
        \path[edge] (t_0_1) to [bend left=10] (t_1_1);
        \path[edge] (t_1_2) to [bend left=10] (t_0_2);
        \path[edge] (t_1_3) to [bend left=10] (t_0_3);
        \end{tikzpicture}
}
\end{center}
Questions of existence and computation of PNEs are fundamental to a game theoretic analysis as a PNE schedule
is stable with respect to selfish players who may try to unilaterally increase their utility by playing a different schedule.

	\begin{theorem}\label{thm:uniform_existence}
		Any \ISG with uniform rewards admits a pure Nash equilibrium which can be computed in polynomial time.
	\end{theorem}

	\begin{proof}[Proof (some details omitted).]
		We iteratively construct a schedule such that every player's schedule is a best response.

		Let 
			$N_i^-(v) := (N_G^-(v) \cup \{v\}) \cap T_i$ denote the closed in-neighborhood of \task $v$ under player $i$'s control, $T_i^{(t)}$ the set of \tasks of player $i$ already scheduled before iteration $t$ and $\alpha_i^{(t)} := |T_i^{(t)}|$.
		In every iteration, we will choose a \task and schedule it together with all remaining \tasks that it depends on. This means that for a \task $v \in T_i^{(t)}$, $a(v)$ is well-defined during iteration $t$.
		We can therefore define
		\begin{equation*}
			\bar\eta^{(t)}_i(v) := \begin{cases}
										\max_{w \in N_i^-(v)} a(w), & N_i^-(v) \setminus T_i^{(t)} = \emptyset\\
										\alpha_i^{(t)} + |N_i^-(v) \setminus T_i^{(t)}|, &\text{else}
									\end{cases}
		\end{equation*}
		Now, $\bar\eta^{(t)}(v) := \max_{i \in N} \bar\eta^{(t)}_i(v)$ represents a tight lower bound for $a(v)$ in any schedule which is a \enquote{completion} of the partial schedule from iteration $t$ (achieved if $v$ and all prerequisites are scheduled immediately).
		
		Similar to Theorem~\ref{th:br-uniform}, it can be verified that player $i$'s schedule $\pi_i$ is a best response if
		for every iteration $t$ and player $i$, the condition \eqref{eq:conflict-free} from Lemma~\ref{th:br-uniform} holds for all \tasks $v,w \in T_i^{(t)}$ and if $v \in T_i^{(t)}\setminus T_i^{(t-1)}$, then $\bar\eta^{(t-1)}(v)$ is minimal among all \tasks from the set $T_i \setminus T_i^{(t-1)}$.
		Furthermore, for every iteration $t$ and services $v \in T_i\setminus T_i^{(t)}$ and $w \in T_i^{(t)}$, we show that $(v,w) \notin E_i$ and $\bar\eta^{(t)}(w) \leq \bar\eta^{(t)}(v)$.
		
		In iteration $t$, we proceed in the following way: Choose a \task $v^*$ that minimizes $\bar\eta^{(t)}$ over all \tasks not yet scheduled and that has no incoming edges from \tasks belonging to the same player. Such a \task must exist, as if $(w,v^*) \in E$ for some \task $w$, then $\bar\eta^{(t)}(w) \leq \bar\eta^{(t)}(v^*)$. Let $i$ be the player such that $v^* \in T_i$.
		
		Assuming that the above conditions are satisfied for iteration $t$, we can now show that 
 %
 %
 		they also hold for iteration $t+1$.
		The described procedure hence constructs a pure Nash equilibrium for the given game in time polynomial in $|T|$.
\end{proof}
	
	As every player strives to activate his \tasks as early as possible, which is also in the interest of other players whose \tasks depend on them, one may think that the schedule that maximizes welfare is always a PNE. 
However, this is not the case. The ratio of the maximum welfare to the maximum welfare in a PNE is called the \emph{price of stability}.  The following theorem shows that this ratio may be strictly greater than 1.
	
	\begin{theorem}
		Even for uniform rewards, a welfare-maximizing schedule is not necessarily a pure Nash equilibrium.
	\label{thm:pos}
	\end{theorem}
	
		\begin{proof}
	Consider the following example.
		\begin{center}
		\scalebox{0.80}{
		\begin{tikzpicture}[]
		\node (p1) {$\pi_1$:};
		\node[vertex] (t_1_1) [right of=p1] {$1$};
		\node[vertex] (t_1_2) [right=1 of t_1_1] {$1$};
		\node[vertex] (t_1_3) [right=1 of t_1_2] {$1$};

		\node (p2) [below of=p1] {$\pi_2$:};
		\node[vertex] (t_2_1) [right of=p2] {$1$};
		\node[vertex] (t_2_2) [right=1 of t_2_1] {$1$};
		\node[vertex] (t_2_3) [right=1 of t_2_2] {$1$};
		
		\node (p3) [below of=p2] {$\pi_3$:};
		\node[vertex] (t_3_1) [right of=p3] {$1$};
		\node[vertex] (t_3_2) [right=1 of t_3_1] {$1$};
		\node[vertex] (t_3_3) [right=1 of t_3_2] {$1$};
		
		\node (p4) [below of=p3] {$\pi_4$:};
		\node[vertex] (t_4_1) [right of=p4] {$1$};
		\node[vertex] (t_4_2) [right=1 of t_4_1] {$1$};
		\node[vertex] (t_4_3) [right=1 of t_4_2] {$1$};

		\path[edge] (t_1_1) to (t_1_2);
		\path[edge] (t_1_2) to (t_2_1);
		\path[edge] (t_1_2) to (t_2_2);

		\path[edge] (t_2_1) to (t_3_2);
		\path[edge] (t_2_1) to  (t_4_2);

		\path[edge] (t_2_2) to (t_3_2);
		\path[edge] (t_2_2) to [bend right=65] (t_4_2);
		
		\path[edge] (t_3_2) to (t_3_3);
		\path[edge] (t_4_2) to (t_4_3);
		\end{tikzpicture}
		}
		\end{center}
		The schedule shown is not a Nash equilibrium: if Player 2 shifts the last \task to the first slot, he increases his reward by 1. In fact, any schedule that is a PNE must have Player 2's last \task (in $\pi_2$ as shown) in the first slot as both other services, depending (by transitivity) on the two \tasks of Player 1 cannot activate before the second time step. Hence, one of the remaining two \tasks of Player 2 (the two with dependencies), that other services depend on, will only be deployed in the last time step.  This implies that in any schedule that is a Nash equilibrium, the two services of both Players 3 and 4 that depend on Player 2's services will not activate before the last time step, either. Hence, Players 3 and 4 cannot achieve a reward higher than $3\cdot 1 + 2\cdot 0 + 1 \cdot (1+1) = 5$ each. Even if both other players receive the maximal reward of 6, then the welfare in any Nash equilibrium schedule cannot exceed 22. On the other hand, the schedule shown achieves a total welfare of 23. Hence, no welfare maximizing schedule can be a Nash equilibrium.
	\end{proof}

Since there may be more than one PNE profile in \ISGs with uniform rewards, it is natural to ask how bad the price of anarchy, the ratio of the maximum welfare schedule to the maximum welfare in a PNE, can become.  

	\begin{theorem}
		The price of anarchy of \ISGs with uniform rewards is $\geq \nicefrac{k(q+1)}{(q+2k-1)}$ with $k$ players, $q$ \tasks each.
	\end{theorem}
	\begin{proof}
		Consider the following example.
		\begin{center}
		\scalebox{0.9}{
		\begin{tikzpicture}[]
		\node (p1) {$\pi_1$:};
		\node[vertex] (t_0_1) [right of=p1] {$1$};
		\node[vertex] (t_0_2) [right of=t_0_1] {$1$};
		\node (t_0_3) [right of=t_0_2] {\dots};
		\node[vertex] (t_0_k) [right of=t_0_3] {$1$};
		\node (p2) [below of=p1] {$\pi_2$:};
		\node[vertex] (t_1_1) [right of=p2] {$1$};
		\node[vertex] (t_1_2) [right of=t_1_1] {$1$};
		\node (t_1_3) [right of=t_1_2] {\dots};
		\node[vertex] (t_1_k) [right of=t_1_3] {$1$};
		\path[edge] (t_0_k) to [bend right=10] (t_1_1);
		\path[edge] (t_0_k) to [bend right=10] (t_1_2);
		\path[edge] (t_0_k) to [bend left=10] (t_1_k);
		\node (p3) [below=0.01 of p2] {$\vdots$};
		\node (pn) [below=0.01 of p3] {$\pi_k$:};
		\node[vertex] (t_n_1) [right of=pn] {$1$};
		\node[vertex] (t_n_2) [right of=t_n_1] {$1$};
		\node (t_n_3) [right of=t_n_2] {\dots};
		\node[vertex] (t_n_k) [right of=t_n_3] {$1$};
		\path[edge] (t_0_k) to [bend left=10] (t_n_1);
		\path[edge] (t_0_k) to [bend left=10] (t_n_2);
		\path[edge] (t_0_k) to [bend left=30] (t_n_k);
		\end{tikzpicture}
		}
		\end{center}
		The worst PNE is obtained (as shown) by scheduling player~1's \task, on which all others depend, at the end; as opposed to the PNE achieved when this \task is at the beginning, which is welfare-maximizing. The ratio between the welfares is
$\frac{ k \cdot q(q+1)/2 }{ q(q+1)/2 + (k-1)q}= \frac{k(q+1)}{q+2k-1}$.
	\end{proof}
	If we fix the number of players $k$, the ratio is bounded by $\lim_{q\rightarrow \infty}k(q+1)/(q+2k-1)=k$. Similarly, when fixing the number of \tasks $q$, then $\lim_{k\rightarrow \infty}k(q+1)/(q+2k-1)=(q+1)/2$.  This motivates the following theorem.

\begin{theorem}
	The price of anarchy of \ISG with uniform rewards is at most $(q+1)/2$. 
\end{theorem}
\begin{proof}
	The worst PNE profile cannot be worse than the schedule in which all \tasks activate at the last time step $q$, which obtains welfare $k\cdot q$. The maximum-welfare schedule cannot be better than a schedule in a game without any precedence constraints, which obtains welfare $k\cdot q(q+1)/2$. Together, we have: $PoA \leq \frac{k q(q+1)/2}{k q}=\frac{q+1}{2}.$
\end{proof}

\subsection{\ISGs with General Rewards}

	Our results for the general setting are not as positive as our results for the uniform rewards setting.  We show that for the general rewards setting, an \ISG with two players does not always admit a pure Nash equilibrium.

	\begin{theorem}\label{thm:no_nash}
        An \ISG with two players and general rewards does not always admit a pure Nash equilibrium.
	\end{theorem}

\begin{proof}

Consider the the following instance.
\begin{center}
\scalebox{0.8}{
\begin{tikzpicture}[]
\node (p1) {$\pi_1$:};
\node[vertex] (t_0_1) [right of=p1] {$1$};
\node[vertex] (t_0_3) [right of=t_0_1] {$4$};
\node[vertex] (t_0_0) [right of=t_0_3] {$3$};
\node[vertex] (t_0_2) [right of=t_0_0] {$2$};
\node (p2) [below of=p1] {$\pi_2$:};
\node[vertex] (t_1_0) [right of=p2] {$2$};
\node[vertex] (t_1_2) [right of=t_1_0] {$4$};
\node[vertex] (t_1_1) [right of=t_1_2] {$1$};
\node[vertex] (t_1_3) [right of=t_1_1] {$3$};
\path[edge] (t_0_1) to [bend left=10] (t_0_3);
\path[edge] (t_1_0) to [bend left=10] (t_0_0);
\path[edge] (t_1_3) to [bend left=10] (t_0_2);
\path[edge] (t_0_3) to [bend left=10] (t_1_2);
\path[edge] (t_1_1) to [bend left=10] (t_1_3);
\end{tikzpicture}
}
\end{center}
\noindent
Assume this game admits a PNE, any best response of Player 1 must satisfy that \task $4$, being the highest reward \task,  is scheduled immediately after \task $1$.  Therefore, any possible best response of Player 1 has to adopt one of the following schedule configurations:
(\emph{i}) $\pi_1 \in (1,4,*,*)$, (\emph{ii}) $\pi_1 \in (*,1,4,*)$ or (\emph{iii}) $\pi_1 \in (*,*,1,4)$.

In a similar way, \task $4$ of Player 2, for any best response of Player 2, must be scheduled as soon as possible. These observations narrow the number of possible PNE configurations to three cases: \textbf{Case $(i)$}
	Player's 2 best response, given any schedule of the form $\pi_1 \in (1,4,*,*)$ is $\pi_2 = (2,4,1,3)$.  However, such a schedule triggers a best response for Player 1 of $\pi_1=(3,1,4,2)$, which take us to case (ii). \textbf{Case $(ii)$}
	Player's 2 best response, given any schedule of the form $\pi_1 \in (*,1,4,*)$ is $\pi_2 = (1,3,4,2)$.  However, such a schedule triggers a best response for Player 1 of $\pi_1=(1,4,2,3)$, which is an instance of case (i). This leads to a cycle of best responses. \textbf{Case $(iii)$}
	Player's 2 best response, given any schedule of the form $\pi_1 \in (*,*,1,4)$ is $\pi_2 \in \{(2,1,3,4),(1,3,2,4)\}$.  However, such schedules trigger a best response for Player 1 of $\pi_1 = (3,1,4,2)$ if $\pi_2 = (2,1,3,4)$, or $\pi_1 = (1,4,3,2)$ in the other case. Both schedules being an instance of case (ii) or (i), respectively.  Therefore, for any schedule $\pi_1$, there is no schedule $\pi_2$, such that $(\pi_1,\pi_2)$ is a PNE.
\end{proof}

We conjecture that the example in Theorem~\ref{thm:no_nash} is minimal with respect to the number of \tasks and dependencies.
We can embed this example into a 3SAT reduction to show that checking the existence of a PNE is NP-hard.

\begin{theorem}\label{thm:existence}
Deciding whether an ISG with general rewards admits a PNE schedule is 
NP-hard, even when each player has at most 4 \tasks.
\end{theorem}

\section{Conclusions}

We have introduced a class of interdependent scheduling games that are motivated by large-scale infrastructure restoration and humanitarian logistics; 
answering many important questions that arise when the players are independent decision makers, including questions of welfare maximization and existence of PNEs.
%
An interesting technical open problem is to determine the complexity of welfare maximization when the number of players is bounded. 
More broadly, there are a number of promising directions for future work including the extension of the model to include cyclic interdependencies \cite{coffrin2012last} or considering other types of manipulation such as adding \tasks or misreporting utilities \cite{ZlotkinR93}.
Also note that approximation algorithms for traditional scheduling settings (with hard dependencies and non-accruing rewards) cannot be directly applied to our model. Hence, another possible avenue of research would be a study of fixed parameter tractability and approximation algorithms for \ISGs.

\clearpage
\section*{Acknowledgments}
Data61/NICTA is funded by the Australian Government through the Department of Communications and the Australian Research Council (ARC) through the ICT Centre of Excellence Program.
Serge Gaspers is the recipient of an ARC Future Fellowship (project number FT140100048) and acknowledges support under the ARC's Discovery Projects funding scheme (project number DP150101134).
Dominik Peters is supported by EPSRC.

\bibliographystyle{named}

\clearpage
\appendix

\section{Full Version of Theorem~\ref{thm:topological}}
\begin{theorem*}
	For one player and general rewards, every welfare-maximizing schedule is a conflict-free schedule.
\end{theorem*}
\begin{proof}
	This follows by an observation about the proof of Lemma~\ref{lem:conflict-free}: In the one-player case, \task $u$ activates immediately under schedule $\pi'_i$ by its maximality among dependencies for which $v$ has to wait. Hence, it also activates immediately under schedule $\pi^*_i$, which is one time step earlier than under schedule $\pi'_i$. Schedule $\pi^*_i$ hence generates strictly more reward than schedule $\pi'_i$.
\end{proof}

\section{Full Version of Theorem~\ref{thm:conflictisgood}}
\begin{theorem*}
	Even if a conflict-free schedule exists, the welfare-maximizing schedule might not be conflict-free.
\end{theorem*}
\begin{proof}
Consider the following example.

\begin{minipage}{0.47\linewidth}
\centering
\scalebox{0.9}{
\begin{tikzpicture}[]
\node (p1) {$\pi^A_1$:};
\node[vertex] (t_1_a) [right of=p1] {$1$};
\node[vertex] (t_2_a) [right of=t_1_a] {$1$};
\node[vertex] (t_3_a) [right of=t_2_a] {$1$};

\node (p2) [below of=p1] {$\pi^A_2$:};
\node[vertex] (t_1_b) [right of=p2] {$1$};
\node[vertex] (t_2_b) [right of=t_1_b] {$100$};
\node[vertex] (t_3_b) [right of=t_2_b] {$100$};

\path[edge] (t_1_a) to [bend left=10] (t_2_a);
\path[edge] (t_2_a) to [bend left=10] (t_3_a);

\path[edge] (t_1_a) to [bend left=10] (t_1_b);
\path[edge] (t_2_a) to [bend left=10] (t_2_b);
\path[edge] (t_2_a) to [bend left=10] (t_3_b);
\end{tikzpicture}
}

$R(\pi^A) = 309$
\end{minipage}
\hfill
\begin{minipage}{0.47\linewidth}
\centering
\scalebox{0.9}{
\begin{tikzpicture}[]
\node (p1) {$\pi^B_1$:};
\node[vertex] (t_1_a) [right of=p1] {$1$};
\node[vertex] (t_2_a) [right of=t_1_a] {$1$};
\node[vertex] (t_3_a) [right of=t_2_a] {$1$};

\node (p2) [below of=p1] {$\pi^B_2$:};
\node[vertex] (t_1_b) [right of=t_2_b] {$1$};
\node[vertex] (t_2_b) [right of=p2] {$100$};
\node[vertex] (t_3_b) [right of=t_2_b] {$100$};

\path[edge] (t_1_a) to [bend left=10] (t_2_a);
\path[edge] (t_2_a) to [bend left=10] (t_3_a);

\path[edge] (t_1_a) to [bend left=5] (t_1_b);
\path[edge] (t_2_a) to [bend left=10] (t_2_b);
\path[edge] (t_2_a) to [bend left=10] (t_3_b);
\end{tikzpicture}
}

$R(\pi^B) = 407$
\end{minipage}

\smallskip
The schedule on the left is conflict-free while the one on the right has a conflict.  Despite the conflict, the right schedule has higher welfare; the two \tasks with reward 100 become active simultaneously in step two, providing more utility to Player 2 and more welfare. 
\end{proof}	

\section{Full Version of Theorem~\ref{th:one-player-np}}

\begin{theorem*}
	For general rewards, \textsc{ISG Welfare} is NP-complete even for a single player.
\end{theorem*}
\begin{proof}
 We give a reduction from the NP-hard problem \textsc{Single machine weighted completion time} \citep{lenstra1978complexity}:  given a set of jobs $J_i \in J$ with each having weight $w_i$, processing time $p_i=1$, and precedence constraints where $i\prec j$ means $J_j$ cannot be scheduled before $J_i$, and integer $k$, is there an ordering of the jobs such that $\sum_{i\in J} w_i C_i \leq k$ where $C_i$ is the completion time of $J_i$?

	
	For each job $J_i\in J$, create \task $t_i$ with reward $r_i=w_i$ and consider the same precedence graph as the one given for jobs. We set $w= (|J|+1) \sum_{i \in J}  w_i - k$.
	
	  By Theorem~\ref{thm:topological}, without loss of generality, we can assume that any schedule for \ISGs  with one player are conflict-free schedules. 
	  It remains to prove that there is an ordering $\pi$ of jobs with a weighted completion time of at most $k$ if and only if the ISG has a conflict-free schedule $\pi'$ with $R(\pi') \geq w$. 
	  
	  Let $\pi=\pi'$, then $C_i$ is the completion time of both, job $J_i$ and \task $t_i$ given ordering $\pi$. Given that $\pi$ is a conflict-free schedule, the contribution of $t_i$ to the objective function is $\left( |T|+1- C_i \right) r_i$. Thus, $R(\pi)= \sum_{i \in T } \left( |T|+1- C_i \right) r_i =  (|T|+1)\sum_{i \in T }r_i - \sum_{i \in T }  r_i C_i$. But, $\sum_{i \in T }  r_i C_i = \sum_{i \in J }  w_i C_i $, which corresponds to the weighted completion time of ordering $\pi$. Therefore, $R(\pi) \geq w \Leftrightarrow \sum_{i \in J }  w_i C_i \leq k$, which concludes the proof. 
\end{proof}

\section{Full Version of Theorem~\ref{thm:uniform_existence}}

	\begin{theorem*}
		Any \ISG with uniform rewards admits a pure Nash equilibrium which can be computed in polynomial time.
	\end{theorem*}

	\begin{proof}
				We iteratively construct a schedule in a way which guarantees that every player's schedule is a best response.

				Let 
					$N_i^-(v) := (N_G^-(v) \cup \{v\}) \cap T_i$
				denote those \tasks controlled by player $i$ that $v$ depends on.
				Denote by $T_i^{(t)}$ the set of \tasks of player $i$ already scheduled before iteration $t$. Let $\alpha_i^{(t)} := |T_i^{(t)}|$ denote the number of such \tasks.
				In every iteration, we will choose a \task and schedule it together with all other (remaining) \tasks that it depends on. This means that for a \task $v \in T_i^{(t)}$, $a(v)$ is well-defined during iteration $t$.
				We can therefore define
				\begin{equation*}
					\bar\eta^{(t)}_i(v) := \begin{cases}
												\max_{w \in N_i^-(v)} a(w), & N_i^-(v) \setminus T_i^{(t)} = \emptyset\\
												\alpha_i^{(t)} + |N_i^-(v) \setminus T_i^{(t)}|, &\text{else}
											\end{cases}
				\end{equation*}
				In particular, observe that if $v$ is controlled by player $i$ and $v \notin T_i^{(t)}$, the second case always applies (as $v \in N_i^-(v)$).

				Furthermore, we define $\bar\eta^{(t)}(v) := \max_{i \in N} \bar\eta^{(t)}_i(v)$ which represents a tight lower bound for the activation time of $v$ in any schedule which is a \enquote{completion} of the partial schedule from iteration $t$ (achieved if all prerequisites are scheduled immediately as the next \tasks). Note that $\bar\eta^{(t)}(v)$ can hence only increase from one iteration to the next and that it reaches the value $a(v)$ as soon as \task $v$ and all its predecessors are scheduled and is constant after that.
		
				By Theorem~\ref{th:br-uniform}, player $i$'s schedule $\pi_i$ is a best response if
				it satisfies condition \eqref{eq:conflict-free} from Lemma~\ref{lem:conflict-free} and for all $v \in T_i$, $\eta(v)$, as defined in Theorem~\ref{th:br-uniform}, is minimal among all \tasks from the set $\{w \in T_i | \pi_i(w) \geq \pi_i(v)\}$.
			
				We will show instead that
				for every iteration $t$ and player $i$, the condition \eqref{eq:conflict-free} from Lemma~\ref{th:br-uniform} holds for all \tasks $v,w \in T_i^{(t)}$ and if $v \in T_i^{(t)}\setminus T_i^{(t-1)}$, then $\bar\eta^{(t-1)}(v)$ is minimal among all \tasks from the set $T_i \setminus T_i^{(t-1)}$.
				To see that this condition is also sufficient for $\pi_i$ being a best response, observe the following: While it may happen for a player $i^*$ and $v,v' \in T_{i^*}$ that $\bar\eta^{(t)}(v)$ is minimal among all \tasks from the set $\{w \in T_i^{(t)} | \pi_i(w) \geq \pi_i(v)\}$ but $\eta(v) > \eta(v')$, this can only occur if for both \tasks the maximum in the definition of $\bar\eta$ is assumed for $i=i^*$ as well as $N_i^-(v) \setminus T_i^{(t)} 0 \{v\}$ and $N_i^-(v') \setminus T_i^{(t)} 0 \{v'\}$.
				This however means that both $v$ and $w$ are equivalent at this point in that they both activate immediately after being deployed.
		
				Furthermore, for every iteration $t$ and \tasks $v \in T_i\setminus T_i^{(t)}$ and $w \in T_i^{(t)}$, we show that $(v,w) \notin E_i$ and $\bar\eta^{(t)}(w) \leq \bar\eta^{(t)}(v)$. This yields that every player's schedule is a best response to the other players' schedules and hence the schedule is in a pure Nash equilibrium.
		
				Assume that the above conditions are satisfied for iteration $t$ and proceed in the following way: Choose a \task $v^*$ that minimizes $\bar\eta^{(t)}$ over all \tasks not yet scheduled and that has no incoming edges from \tasks belonging to the same player. Such a \task must exist, as if $(w,v^*) \in E$ for some \task $w$, then $\bar\eta^{(t)}(w) \leq \bar\eta^{(t)}(v^*)$. Let $i$ be the player such that $v^* \in T_i$.
				If $v^*$ has no incoming edges from any of the \tasks not yet scheduled, then scheduling $v^*$ as the next \task of player $i$ satisfies the best-response criterion, no matter the ordering of the unscheduled \tasks.
		
				Hence, suppose that $v^*$ depends on some other \tasks not yet scheduled.
				Denote this set of \tasks by $S$. By induction, scheduling all \tasks in $S$ (respecting the ordering required by edges in $E_i$ if necessary) satisfies condition \eqref{eq:conflict-free} for all players $i$ and $v,w \in T_i^{(t+1)}$.
				Furthermore, note that for every $w \in S$, $\bar\eta^{(t)}(w)=\bar\eta^{(t)}(v^*)$ by minimality of $v^*$ and the dependency of $v^*$ on $w$, thus $\bar\eta^{(t)}(w)=\bar\eta^{(t)}(v^*)$.
				Hence for every player $i$, if $v \in T_i^{(t+1)}\setminus T_i^{(t)}(\subseteq S \cup \{v^*\})$, then $\bar\eta^{(t)}(v)$ is minimal among all \tasks from the set $T_i \setminus T_i^{(t)}$.
		
				Finally, the criteria for every $v \in T_i \setminus T_i^{(t+1)}$ and $w \in T_i^{(t+1)}$ are satisfied as well:
				For every $v \notin S$ and $w \in S$, $(v,w) \notin E_i$ as otherwise $w\in S$. Furthermore, $\bar\eta^{(t+1)}(w) = \bar\eta^{(t)}(w) = \bar\eta^{(t)}(v^*) \leq \bar\eta^{(t)}(v) \leq \bar\eta^{(t+1)}(v)$ where the first equality holds because $w$ and all its dependencies are scheduled in iteration $t$, the second equality was shown above and the inequalities follows by minimality of $v^*$ and monotonicity of $\bar\eta^{(t)}(v)$ in t.
		
				The described procedure hence constructs a pure Nash equilibrium for the given game in time polynomial in $|T|$.
	\end{proof}

\section{Full Version of Theorem~\ref{thm:existence}}

\begin{theorem*}
Deciding whether there exists a PNE schedule is 
NP-hard, even when each player has at most 4 \tasks.
\end{theorem*}
\begin{proof}
 We give a reduction from the NP-hard problem \textsc{3SAT} \citep{karpreducibility}: given a CNF formula $F$ where each clause contains exactly 3 literals, is there an assignment to the variables of $F$ such that all clauses are satisfied?

			 
For each variable $x$ in $F$, create a player $P_x$ with \tasks $T_x =\{x,\neg x\}$. Both \tasks have the same reward $r(x)=r(\neg x)=1$. For each clause $c$ in $F$, create a player $P_c$ with \tasks $T_c= \{c_1,c_2,c_3,d_c\}$ and set rewards to be 
$r(d_c)=3$ and $r(c_1)=r(c_2)=r(c_3)=4$. For each clause $c$, we create a gadget $G_c$ corresponding to a copy of the \ISG from Theorem \ref{thm:no_nash} which admits no PNE and consists of 2 players with 4 \tasks each. For each clause $c=(\ell_1 \vee \ell_2 \vee \ell_3)$ in $F$, the precedence graph contains arcs $(\ell_1,c_1)$,$(\ell_2,c_2)$,$(\ell_3,c_3)$ and arcs from \task $d_c$ to the 8 \tasks of gadget $G_c$.
	
	It remains to prove that $F$ has an assignment satisfying all clauses if and only if the \ISG  admits a pure Nash equilibrium. For the forward direction,  suppose $F$ has an assignment $\alpha: \mbox{var}(F)\rightarrow \{0,1\}$ satisfying all clauses. Consider the schedule where, for each variable $x$, the player $P_x$ schedules first the literal of $x$ that is set to true by $\alpha$, i.e., $x$ is scheduled before $\neg x$ iff $\alpha(x)=1$. For each clause $c$, the player $P_c$ schedules its true literals, then its false literals given $\alpha$, and then \task $d_c$. \Tasks in gadget $G_c$ can be scheduled arbitrarily.
	This schedule is in a pure Nash equilibrium: for each variable player this is the best that player can do. For each clause player this is the best that player can do given that all dependencies from the \tasks of the variable players are met. Finally, the players in gadget $G_c$ are indifferent between all schedules because their \tasks all become active in the last time step, given that \task $d_c$ was scheduled at the end.
	
	For the reverse direction, suppose conversely that the game has a pure Nash equilibrium. Consider the assignment $\alpha: \mbox{var}(F)\rightarrow \{0,1\}$ with $\alpha(x)=1$ iff player $P_x$ schedules $x$ at the first time step.
We show that the assignment $\alpha$ satisfies $F$. Suppose some clause $c$ is not satisfied. Then, none of its literal \tasks will be activated before the second time step, only \task $d_c$ is activated in the first time step. Hence, all best responses for the clause player $P_c$ put \task $d_c$ into the first time slot, giving the player a reward of 36.
This means that \tasks in gadget $G_c$ have no restrictions imposed. But $G_c$ for itself does not admit a Nash equilibrium, and hence the entire game does not either, a contradiction. So all clauses are satisfied.
\end{proof}

\end{document}